\documentclass[12pt,a4paper]{article}

\usepackage{amsthm}
\usepackage{amssymb}
\usepackage{amsmath}
\usepackage{graphicx}
\usepackage{pdfpages}
\usepackage{url}

\newcommand\SNR{\mathrm{SNR}}
\newcommand\E{\mathbb{E}}

\newtheorem{theorem}{Theorem}
\newtheorem{proposition}{Proposition}
\newtheorem{lemma}{Lemma}
\newtheorem{corollary}{Corollary}

\begin{document}

\title{A proposal for informative default priors scaled by the standard error of estimates}
\author{Erik van Zwet\footnote{Department of Biomedical Data Sciences, Leiden University Medical Center.} \qquad Andrew Gelman\footnote{Department of Statistics and Department of Political Science, Columbia University. This work was partially supported by U.S. Office of Naval Research.}}
\date{\today}

\maketitle

\begin{abstract}
If we have an unbiased estimate of some parameter of interest, then its absolute value is positively biased for the absolute value of the parameter. This bias is large when the signal-to-noise ratio (SNR) is small, and it becomes even larger when we condition on statistical significance; the winner's curse. This is a frequentist motivation for regularization. To determine a suitable amount of shrinkage, we propose to estimate the distribution of the SNR from a large collection or corpus of  similar studies and use this as a prior distribution. The wider the scope of the corpus, the less informative the prior, but a wider scope does not necessarily result in a more diffuse prior. We show that the estimation of the prior simplifies if we require that posterior inference is equivariant under linear transformations of the data. We demonstrate our approach with corpora of 86 replication studies from psychology and 178 phase 3 clinical trials.  Our suggestion is not intended to be a replacement for a prior based on full information about a particular problem; rather, it represents a familywise choice that should yield better long-term properties than the current default uniform prior, which has led to systematic overestimates of effect sizes and a replication crisis when these inflated estimates have not shown up in later studies.
\end{abstract}


\section{Introduction}
Regression modeling plays a central role in the biomedical and social sciences. Linear and generalized linear models, generalized estimating equations, and quantile regression offer great flexibility and are easy to use. When the sample size is not too small,  statistical inference can be based on the fact that M-estimates of regression coefficients are approximately normal and unbiased \cite{stefanski2002calculus}. This yields the familiar frequentist inference in terms of normality-based confidence intervals and $p$-values, and it also leads to informative Bayesian approaches in which the unbiased estimates form a likelihood which can be augmented with hierarchical models and other forms of prior distribution.

If we have an unbiased estimate of some parameter of interest, such as a regression coefficient, then by Jensen's inequality its absolute value is positively biased for the absolute value of the parameter. This bias is large when the signal-to-noise ratio (SNR) is small, and it becomes even larger when we condition on statistical significance. This is called the winner's curse or type M error \cite{gelman2000type,ioannidis2005contradicted,gelman2014beyond}. We conclude that noisy estimates must be regularized or partially pooled toward zero. However, the degree of this shrinkage should be carefully considered. Too little shrinkage means that we will systematically overestimate effect sizes, which then later do not replicate. On the other hand, too much shrinkage could lead to missing real discoveries. 

From the Bayesian perspective, the answer to ``how much shrinkage'' depends on the prior.  Here, we propose to obtain the relevant prior information from a large collection or corpus of  similar studies. Such a prior can then be used for default or routine Bayesian inference. The wider the scope of the corpus, the less informative the prior and the more generally applicable. Moreover, a wide scope allows us to include many studies in the corpus so that we can estimate the prior information accurately. Perhaps the most important point we want to make is that a wide scope does not necessarily result in a more diffuse prior. 

In the next section, we will motivate the present paper by discussing in more detail why noisy estimates must be regularized. Then, we argue that we can determine the suitable amount of shrinkage by estimating the distribution of the signal-to-noise ratio (SNR) in a particular area of research. We show that a particular independence assumption will make the estimation easier, and ensures that the posterior inference is unaffected by trivial changes of measurement unit. We also show that depending on the shape of the distribution of the SNR, the amount of shrinkage will be adaptive to the SNR.  We demonstrate our approach with two examples. We end the paper with a discussion.

\section{Exaggeration and the need to shrink}

We will ignore small sample issues by assuming that we have a normally distributed, unbiased estimate $b$ of a regression coefficient $\beta$ with known standard error $s$. In other words, conditionally on $\beta$ and $s$, $b$ has the normal distribution with mean $\beta$ and standard deviation $s$ and therefore
\begin{equation}\label{conf}
\mbox{Pr}(b \in [\beta \pm 1.96 \, s] \mid \beta, s)=0.95.
\end{equation}
This setup may appear to be overly simplistic, but inference about regression parameters based on Wald type confidence intervals (and associated $p$-values) is common practice and, as noted above, this is also the building block for the standard Bayesian approach:  even if you don't care about unbiased estimation per se, independent normally distributed unbiased estimates can be used to construct a likelihood function. Exact intervals based on the $t$ distribution are available in linear models, but the difference is already very small in most real-world examples.

Since $b$ is unbiased for $\beta$, it follows from Jensen's inequality that $|b|$ is positively biased for $|\beta|$. This bias is large when $b$ is ``noisy,'' i.e.\ when the signal-to-noise ratio (SNR) $|\beta|/s$ is small. The bias becomes even larger when we condition on statistical significance, which is called the winner's curse. The relation between overestimation of the effect size and the SNR has been demonstrated through simulation \cite{gelman2014beyond,ioannidis2008most}. More recently, the following theorem has been established \cite{vanZwetCator}.

\begin{theorem}
Suppose $b$ is normally distributed with mean $\beta$ and standard deviation $s$. For every $c>0$, the exaggeration ratio,
$$\E\left(|b/\beta|  \mid s, \beta, |b|>c\right) $$
depends on $\beta$ and $s$ only through the SNR $|\beta|/s$. The exaggeration ratio is always greater than 1. It is decreasing in the SNR and increasing in $c$.
\end{theorem}

\noindent
The exaggeration ratio is also known as the type M error \cite{gelman2014beyond}.

A partial solution to this overestimation of effect size is to use ``weakly informative'' priors \cite{gelman2008weakly, greenland2015penalization}, but then the question arises:  how informative should the priors be?  The literature on weakly informative priors tends to focus on superior performance compared to noninformative priors. Here we propose to obtain realistic but general prior information from large collections or ``corpora'' of similar studies. Such priors can be used for default or routine Bayesian inference. The priors we propose can be narrow and result in a considerable degree of shrinkage.

\section{Constructing a default informative prior}

\subsection{Using a corpus of previous studies}

Researchers in the life sciences often believe that they have little or no prior information because their study is unique; nobody has ever studied that particular intervention or exposure in that particular population with that particular outcome under those particular circumstances. We believe that it is a mistake to think like that. At the highest level of aggregation, just knowing that you are doing another study in the domain of the life sciences represents a lot of information. 

It is often possible to be more specific, but that does involve making choices that depend on the details of the study in question. The more we zoom in, the smaller the set of relevant examples becomes. This will make it harder to determine the prior distribution accurately. So, we propose to estimate prior distributions from large, broad collections of studies. Obtaining prior information using external data is sometimes referred to as empirical Bayes \cite{carlin2010bayes, efron2012large}.  

For our purposes here, we define a  {\em corpus} as a collection of pairs $(b_j,s_j)$ from studies $j$ that are similar in the sense that they meet certain inclusion criteria. An example would be placebo-controlled phase 3 randomized clinical trials (RCTs). If we know only that a particular study meets a set of inclusion criteria---or we are willing to ignore all other features of that study---then we can model that study exchangeably with all others that meet those criteria. The inclusion criteria of the corpus represent exactly the information that we are including in the prior. This implies that making the scope wider by removing certain criteria means putting less information into the prior. 

There is no reason to expect that fewer inclusion criteria would yield a more widely dispersed prior distribution. Indeed, the prior can become less dispersed if by widening the scope, we add  mostly studies with small effects. From this point of view, one cannot tell by looking at the prior distribution how much substantive information it represents. This is a different point of view from the usual conviction that high-variance priors carry little information \cite{kass1996selection}.

Obtaining prior information from corpora with a wide scope has two important practical advantages. First, the wide scope ensures that many studies are eligible so that the prior can be estimated accurately. Second, the resulting prior can be used as a default in a wide range of applications.

\subsection{Estimating the prior}
Consider three random variables:  $\beta$, $b$, and $s$, with joint density $p$ in the defined population. We assume that $b$ has the normal distribution with mean $\beta$ and standard deviation $s$. We observe $b$ and $s$, and we want to do Bayesian inference about  $\beta$.  Conditionally on $s$, Bayes' rule states
\begin{equation}
p(\beta \mid b,s) \propto p(b \mid \beta,s)p(\beta \mid s),
\end{equation}
where $p(b \mid \beta,s)$ is the conditional density of $b$ given $\beta$ and  $s$, and $p(\beta \mid s)$ is the conditional density of $\beta$ given $s$. 

To be able to do Bayesian inference about $\beta$, our goal is to estimate the probability kernel $p(\beta \mid s)$ on the basis of a sample of observed pairs  $(b_j,s_j)$. Since we are assuming that $p(b \mid \beta,s)$ is $\mbox{normal}(\beta,s)$, it is actually possible to estimate the full joint distribution of $(\beta,b,s)$ from just the pairs $(b_j,s_j)$.

Estimating the conditional distribution of $\beta$ given $s$ is equivalent to estimating the conditional distribution of the signal-to-noise ratio $\beta/s$ given $s$. The $z$ value $b/s$ is the sum of $\beta/s$ and an independent standard normal random variable. So, we can first do some regression modeling to estimate the conditional distribution of $b/s$ conditional on $s$ and then do a deconvolution to obtain conditional distribution of  $\beta/s$ given $s$. Finally, we obtain the prior distribution of $\beta$ given $s$ by scaling.

Suppose we succeed in accurately estimating $p(\beta \mid s)$ from a large corpus of exchangeable pairs $(b_j,s_j)$. If we use that estimate as a prior, then the posterior will be approximately calibrated with respect to that corpus. That is, posterior probabilities will represent frequencies across the corpus.

\subsection{Independence assumption}

Our goal is to develop priors that are widely applicable for routine Bayesian inference. In that context, it is desirable that inference about $\beta$ does not depend on trivial data transformations, such as switching events and non-events in logistic regression, relabeling dummies, changing the unit of measurement of covariates or a numerical outcome.  Mathematically, such a requirement means that, 
\begin{equation}\label{requirement}
p(\beta \mid b, s) =|c|\, p(c \beta \mid cb, |c|s), \mbox{ for all } c\neq 0.
\end{equation}
Requirement (\ref{requirement}) means that posterior inference about $\beta$ is equivariant under linear transformations of the data. It has the following interpretation in terms of the SNR $\beta/s$.

\begin{theorem}
Requirement (\ref{requirement}) holds if and only if
\begin{enumerate}
\item $s$ and $\beta/s$ are independent, and
\item the distribution of $\beta/s$ is symmetric around zero.
\end{enumerate}
\end{theorem}

Requirement (\ref{requirement}) drastically simplifies the estimation of the probability kernel  $p(\beta \mid s)$, because we need only estimate the {\em marginal}  density of $\beta/s$. Moreover, we may assume this density is symmetric. We only have to estimate the symmetric marginal density of $b/s$ and then do a deconvolution to obtain the marginal density of $\beta/s$. We can then get the distribution of $\beta$ given $s$ by simple scaling.

We motivated (\ref{requirement}) from a pragmatic point of view by insisting that posterior inference about $\beta$ should be equivariant under linear transformations to avoid cheating.  However, (\ref{requirement}) can also be interpreted as an assumption about the joint distribution of the observables $b$ and $s$. Since $b/s$ is the sum of $\beta/s$ and an independent standard normal random variable, a trivial consequence of Theorem 2 is
\begin{corollary}
Requirement (\ref{requirement}) hold if and only if
\begin{itemize}
\item $s$ and $b/s$ are independent, and
\item the distribution of $b/s$ is symmetric.
\end{itemize}
\end{corollary}
We can check if these properties hold, at least to reasonable approximation, in any particular corpus. An important necessary condition for the above to hold, is that $s$ and $|b|$ are positively correlated.  We argue from an anthropic principle \cite{gelman2018anthropic} that it is reasonable to expect such a correlation, as follows. Studies are commonly designed to have just enough power so that effects can just about be estimated from data. Indeed, the goal of sample size calculations (formal or informal) is to balance $|b|$ and $s$ so that the probability that $|b|/s$ exceeds 1.96 is not too large or too small. Hence effects tend to be of the same order of magnitude as standard errors. This does not preclude that the distribution of the SNR can differ between corpora. For example, some research areas might have larger effects, better measurement devices or more funding opportunities for large studies.

If (\ref{requirement}) holds then the prior information about the SNR is all we need for estimating $\beta$. Moreover, we have the following equality for the posterior mean
\begin{equation}
\E(\beta \mid b,s) = s\, \E(\SNR \mid z).
\end{equation}
If (\ref{requirement}) does not hold then this equality doesn't hold either, but we claim that it is still to sensible use the shrinkage estimate $s\,\E(\SNR \mid z)$. By conditioning on $z$, we are using the prior information about the SNR. So, as far as shrinkage is concerned, we are using all the relevant information.

\subsection{Adaptivity and consistency with a $t$ prior}
Suppose we have a normally distributed, unbiased estimator $b_n$ of $\beta$ with known standard error $s_n=\sigma/\sqrt{n}$, where $n$ is the sample size. If we consider $b_n$ to be the full data, we can combine it with a prior for $\beta$ to perform Bayesian inference. If we choose a fixed prior for $\beta$, then its influence disappears as the sample size increases in the sense that the posterior distribution of $\beta$ converges to the likelihood of $b_n$. In particular, the posterior mean of $\beta$ converges to $b_n$ and hence is a consistent estimate of $\beta$. This is a special case of the well-known Bernstein--von Mises theorem. Here, we are proposing to use a fixed prior for the signal-to-noise ratio $\beta/s_n$. Thus, the implied (scaled) prior for $\beta$ depends on the sample size, and therefore Bernstein--von Mises does not apply.

For example, if we put a normal prior with mean zero and standard deviation $\tau$ on $\beta/s_n$, then the posterior mean for $\beta$ is $\frac{\tau^2}{\tau^2 +1} b_n$. Evidently, this is an inconsistent estimate of $\beta$, unless $\beta$ happens to be zero.

Fortunately, by choosing a prior with flatter tails than the normal, it is possible to have a fixed prior on $\beta/s_n$ and still have the posterior distribution of $\beta$ converging to the likelihood of $b_n$. The following theorem is a special case of a result due to Dawid in the context of Bayesian outlier detection \cite{dawid1973posterior}.

\begin{theorem}
Suppose we have a normally distributed, unbiased estimate $b_n$ of $\beta$  with known standard error $s_n=\sigma/\sqrt{n}$, where $n$ is the sample size, and suppose $\beta$ is assigned a sample-size-dependent $t_{\nu}(0,s_n)$ prior distribution. Then, as long as the true $\beta$ is not equal to zero, the limiting posterior distribution of $(\beta - b_n)/s_n$ is standard normal. 
\end{theorem}

The point is that the $t$ distribution has a much flatter tail than the normal distribution. As $n$ becomes large, the likelihood of $b_n$ will concentrate around the true and nonzero $\beta$.  Meanwhile, the prior, by construction, becomes increasingly narrow and is centered around 0. Thus the overlap with the normal likelihood will be in a region where the prior is almost completely flat and hence the posterior will converge to the likelihood. In other words, as $n$ grows and the $z$-value $b_n/s_n$ becomes large, the shrinkage disappears.  In that sense, the amount of shrinkage adapts to the signal-to-noise ratio.

Dawid's theorem is actually more general and provides sufficient conditions for the tail behavior of the prior.  There is an extensive literature about heavy-tailed priors which is reviewed by O'Hagan and Pericchi \cite{o2012bayesian}.

\subsection{Mixture of normals prior}

Above we established that using a $t$ prior distribution for $\beta/s$ yields a consistent estimator. In practice, we prefer to use a finite mixture of zero-mean normal distributions. This has two advantages. First, all calculations can be done explicitly, which is fast and can give us insight. The mathematical details are not difficult, and we describe them in the appendix. More importantly, a mixture of zero-mean normal distributions is a very flexible model. Already with just two components, we can separately fit the central part and the tails of the distribution of $\beta/s$. As it turns out, a mixture of two components provides a reasonably good fit in our examples, so this is what we used there.

The tails of such a mixture are Gaussian and hence we will not get consistent estimates. However, this is not a major concern. Sample sizes never actually go to infinity, and if one of the components has a large standard deviation, then the tails of the mixture will be flat enough for practical purposes.
As with many statistical models (for example, logistic vs.\ probit regression), what is most important is not the exact functional form but rather that the model has enough flexibility that we can learn from data.

\section{Example using corpora in psychology and medicine}\label{meta-analyses}
We will illustrate the ideas of this paper with two example corpora, one from psychology and one from medicine.

To obtain reliable prior information, the reported effects in our corpus must be a fair sample of the population of effects within the scope. It is well-known that for various reasons (publication bias, file drawer effect, researcher degrees of freedom, fishing, forking paths, etc.)\ reported effects tend to be inflated \cite{ioannidis2005contradicted, rothstein2006publication, ioannidis2008most, button2013power, gelman2013garden, open2015estimating}.  Here, we consider two special cases where we expect to find a reasonably honest sample of effects.

\subsection{Open Science Collaboration study on reproducibility in psychology}
To assess the reproducibility of psychological science, the Open Science Collaboration (OSC) selected 100 studies from three leading journals of the American Psychological Association, and replicated them \cite{open2015estimating}. They chose the journals Psychological Science, Journal of Personality and Social Psychology, and Journal of Experimental Psychology: Learning, Memory, and Cognition. According to the authors of the replication study, the first journal is a premier outlet for all psychology research; the second and third are leading disciplinary-specific journals for social psychology and cognitive psychology, respectively. The studies were selected in a quasi-random way to balance two competing goals:  ``minimizing selection bias by having only a small set of articles available at a time and matching studies with replication teams' interests, resources, and expertise.'' The data are publicly available at {\tt https://osf.io/fgjvw/}. Our analyses performed are available as a supplement to this paper.

The regression parameters and their standard errors are not available in this data set, so we transformed the two-sided $p$-values of the replication studies to absolute $z$-values by  
$$|z| = -\,\Phi^{-1}(p\text{-value}/2).$$
Under requirement (\ref{requirement}) the absolute $z$-values are sufficient to estimate the prior. Excluding the $F$-tests, we have 86 absolute $z$-values.

We estimate the distribution of the SNR as a mixture of two zero-mean normal distributions with standard deviations $\tau_1=0.7$ and $\tau_2=4.0$ and mixture proportions 0.57 and 0.43, respectively. We show the marginal distribution of the SNR in Figure \ref{fig:snr}. We show the shrinkage factor in Figure \ref{fig:shrinkage} as a function of the $z$-value. The shrinkage factor is more than 2 for small values of $z$. At $z=1.96$ the shrinkage factor is 1.7. In Figure \ref{fig:sign}, we show the posterior probability that the sign of $\beta$ differs from the sign of $b$. At $z=1.96$ this probability is about 9\%.

\begin{figure}[htp] \centering{
\includegraphics[scale=0.8]{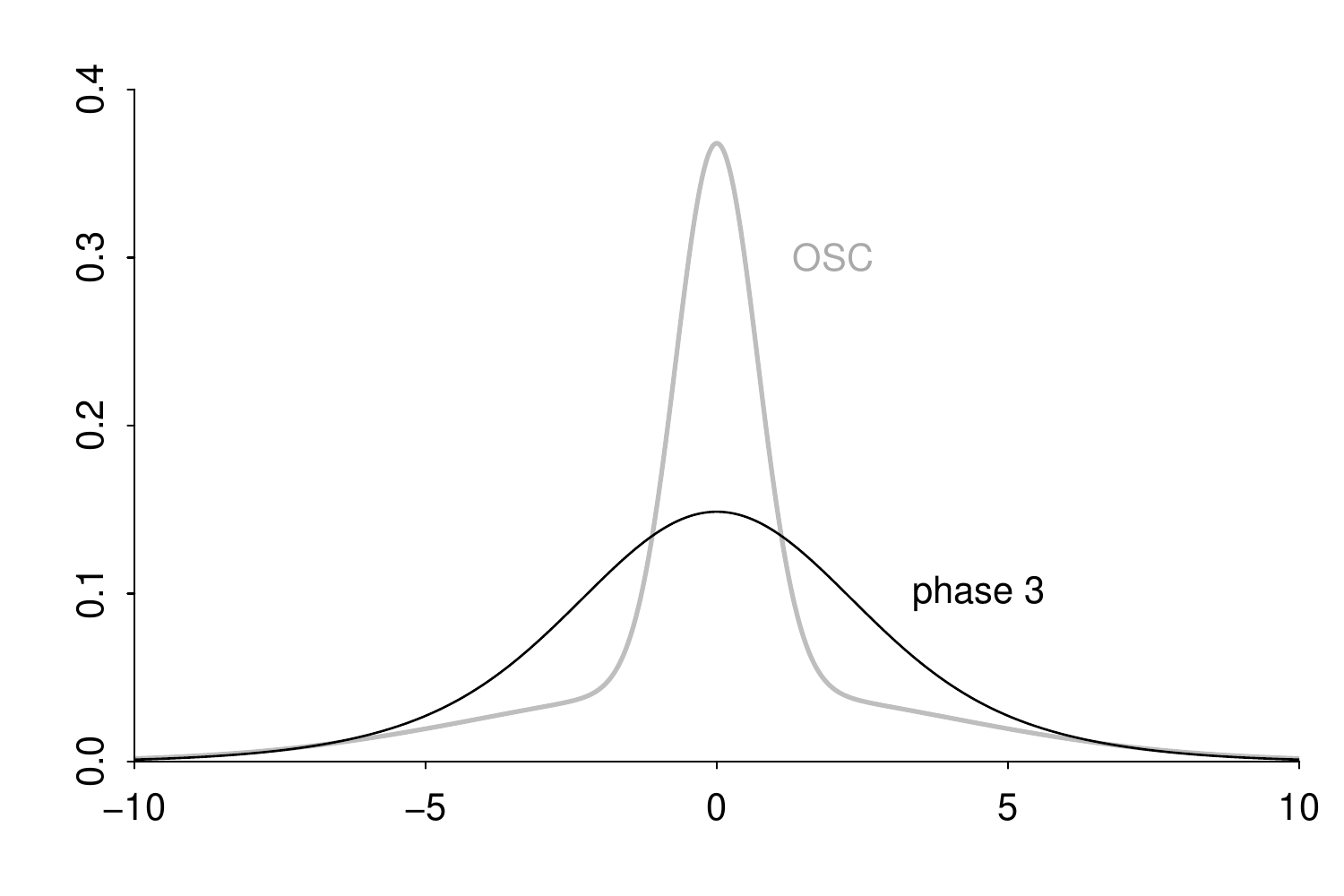}}
\caption{\em Estimated prior distribution for the signal-to-noise ratio, $b/s$, in each of the two classes of problems considered in this article:  psychology studies from the Open Science Collaboration, and Phase 3 clinical trials in medicine. For both cases we fit mixtures of two zero-centered normals..}\label{fig:snr}
\end{figure}

\begin{figure}[htp] \centering{
\includegraphics[scale=0.8]{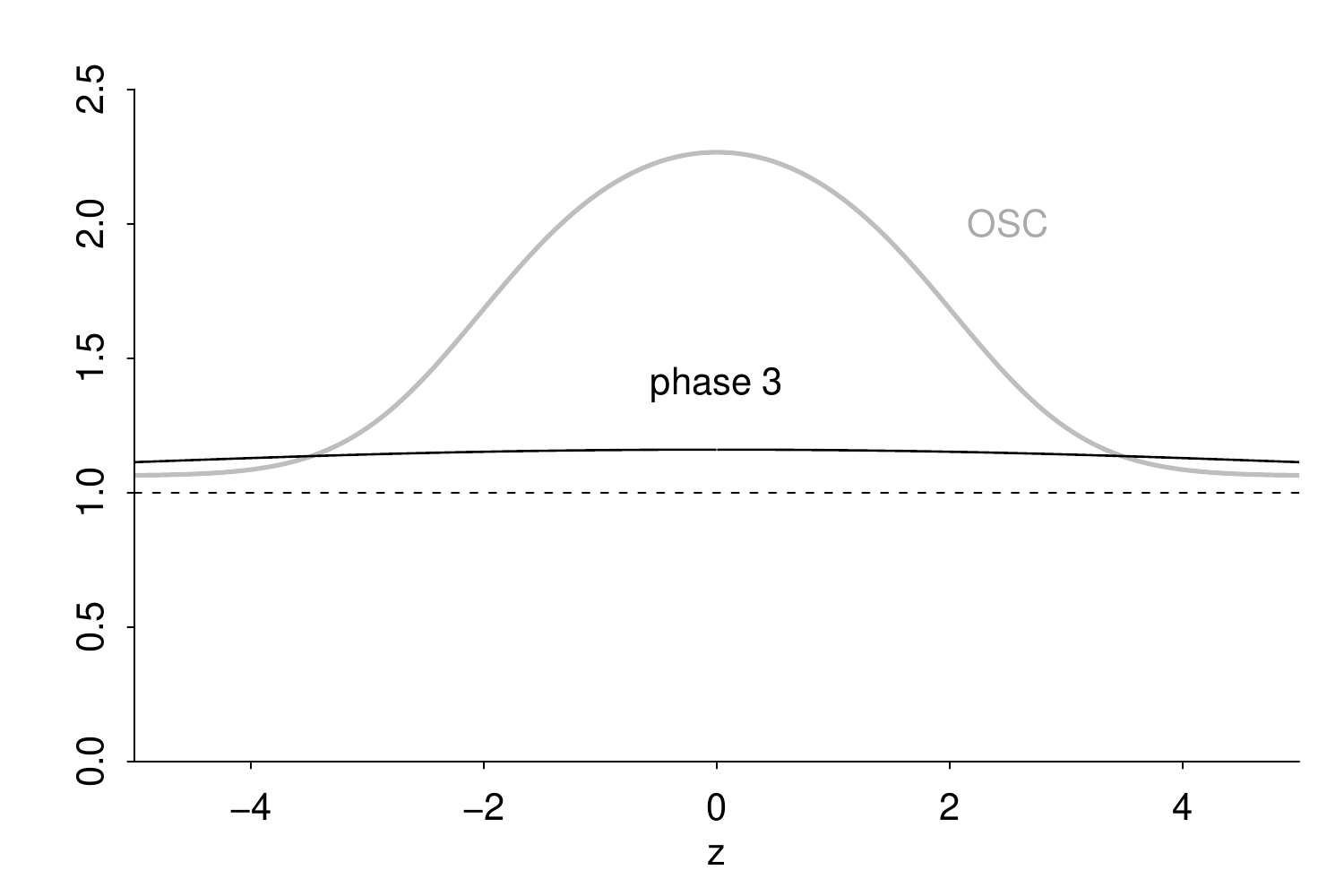}}
\caption{\em Shrinkage factor---the amount by which the raw estimate is divided to yield the Bayes estimate---as a function of the $z$-score of a new study, for each of our two classes of problems.  A shrinkage factor of 1 corresponds to no shrinkage.}\label{fig:shrinkage}
\end{figure}

\begin{figure}[htp] \centering{
\includegraphics[scale=0.8]{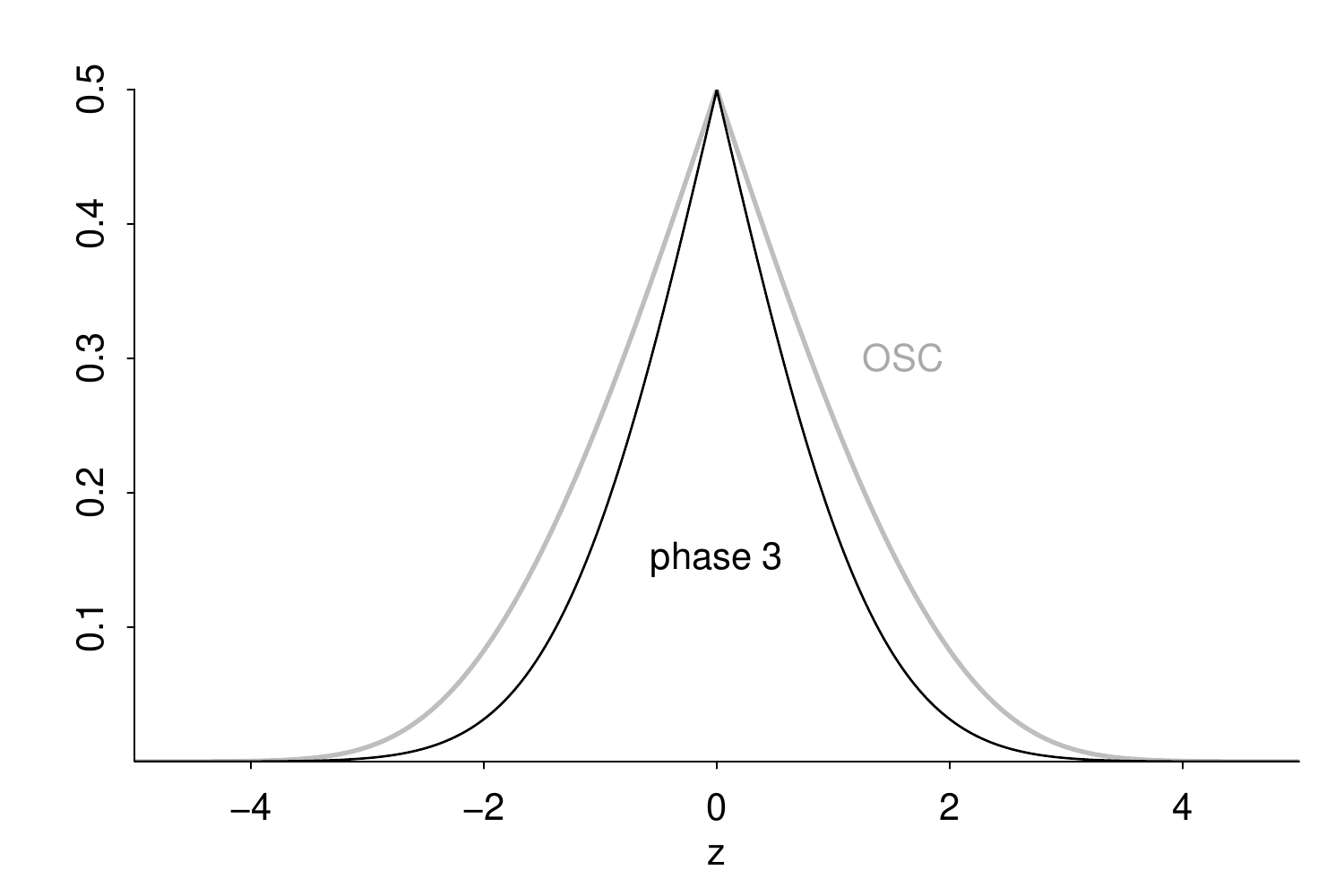}}
\caption{\em Posterior probability of a sign error (that is, the probability that the true effect $\beta$ has the opposite sign of the observed estimate $b$) as a function of the $z$-score, for each of our two classes of problems considered in this article.}\label{fig:sign}
\end{figure}

\subsection{Cochrane phase 3 placebo controlled clinical trials}

The Cochrane Database of Systematic Reviews (CDSR) is the leading journal and database for systematic reviews in health care. We downloaded all primary studies up to 2018 and selected those where either the title or the Cochrane's methods description contained the terms ``phase 3'' or ``phase III.'' Next, we selected all comparisons for efficacy against a placebo. For each study, we selected only a single $z$-value where we tried to obtain the effect of primary interest. We were left with 178 $z$-values.

We estimate the distribution of the SNR as a mixture of two zero-mean normal distributions with standard deviations $\tau_1=2.1$ and $\tau_2=3.6$ and mixture proportions 0.48 and 0.52, respectively. We show the marginal distribution of the SNR  in Figure  \ref{fig:snr}. We show the shrinkage factor in Figure \ref{fig:shrinkage} as a function of the $z$-value. The shrinkage factor is about 1.2 for small values of $z$. At $z=1.96$ the shrinkage factor is 1.15. In Figure \ref{fig:sign}, we show the posterior probability that the sign of $\beta$ differs from the sign of $b$. At $z=1.96$ this probability is about 3\%.

\subsection{Remarks}

The two analyses we have performed yielded very different results. Unsurprisingly, the signal-to-noise ratio in psychology research tends to be much smaller than in phase 3 clinical studies. The latter usually involve very considerable investments and hence may be expected to have high statistical power. Consequently, it seems that one should apply much stronger shrinkage to results from psychological research than from phase 3 clinical studies, especially when the observed absolute $z$-value is small.

An important conclusion of the OSC reproducibility study was that on average the effect size of the replication effects was half the magnitude of the effect size of the original effects \cite{open2015estimating}. This roughly agrees with our analysis. Using {\em only} the results of the replication studies, we found that shrinkage by a factor of about 1.5 to 2 is typically in order.

A more extensive analysis of the entire Cochrane database is reported elsewhere \cite{vanZwetSenn}.

\section{Discussion}

\subsection{The value of default informative priors}

Nearly forty years ago, Donald Rubin wrote \cite{rubin1984bayesianly}:

\begin{quotation}\noindent
Another reason for the applied statistician to care about  Bayesian inference is that consumers of statistical answers, at least interval estimates, commonly interpret them as probability statements about the possible values of parameters. Consequently, the answers statisticians provide to consumers should be capable of being interpreted as approximate Bayesian statements.
\end{quotation}
The present paper is an attempt to do just that. The confidence interval  (\ref{conf}) describes the long-run coverage performance of the random interval $[b - 1.96\, s ,b + 1.96 \, s]$. The statement does not hold conditionally on the data, but it is often mistakenly interpreted ``Bayesianly'' as 
\begin{equation}\label{prob}
\mbox{Pr} (\beta\in [b\pm  1.96\, s] \mid b, s)=0.95.
\end{equation}
where $\beta$ is viewed as a random variable, and we condition on the data pair $(b,s)$. We refer to Greenland et al.\ \cite{greenland2016statistical} for a discussion of this misinterpretation. Statement (\ref{prob}) is arguably more relevant than (\ref{conf}) because it refers to the data at hand, rather than the procedure being used. This may explain, at least in part, the pervasiveness of the misinterpretation; it is what researchers want to know.

The Bayesian statement (\ref{prob}) is only valid if $\beta$ has the (improper) uniform or ``flat''  prior distribution. The matching property of (\ref{conf}) and (\ref{prob}) has lead many to consider the uniform prior to be an objective or noninformative prior \cite{ghosh2011objective}. Many other criteria have been proposed by which a priori might be considered to be objective \cite{kass1996selection}, but in the normal location model with known standard deviation they all yield the (improper) uniform distribution as the unique objective prior.  So, we find that the flat prior is used  for Bayesian inference about regression coefficients in two distinct situations: explicitly with the goal of objective Bayesian inference and implicitly whenever the confidence interval for a regression coefficient is interpreted as a credibility interval. 

The goal of using a noninformative prior is to be impartial or objective by minimizing the influence of the prior on the posterior, see \cite{berger2006case} but also \cite{gelman2017beyond}. However, this influence depends on which aspect of the posterior we are considering.  The flat prior is actually very informative both for the magnitude and the sign of $\beta$. This is just a consequence of the fact that a diffuse prior favors large absolute values. In fact, use of the flat prior results in overestimation of the magnitude of $\beta$ and exaggerated evidence about its sign \cite{van2018default}: type M (magnitude) and type S (sign) errors \cite{gelman2014beyond}. We thus echo the classical Bayesian literature in concluding that ``noninformative prior information'' is a contradiction in terms. The flat prior carries information just like any other; it represents the assumption that the effect is likely to be  large. This is often not true. Indeed, the signal-to-noise ratio $\beta/s$ is often very low and then it is necessary to shrink the unbiased estimate. Failure to do so by inappropriately using the flat prior causes overestimation of effects and subsequent failure to replicate them.  

Some degree of shrinkage is achieved by using weakly informative priors. This can provide good results in many situations \cite{gelman2008weakly, greenland2015penalization}, but can still lead to undershrinkage and positively biased effect size estimates.  Here we propose to use prior information estimated from a large corpus of similar studies, under the constraint that we are modeling effect size in standard error units. By using a wide scope, we can ensure that little information is required that is specific to the study at hand. A wide scope also means that we can include many studies so that  the prior can be estimated accurately. Finally, a wide scope means that the prior information is applicable to many studies. A wide scope does {\em not} imply that the prior will be wide in the sense of having high variance.

If we succeed in accurately estimating the prior information from a large corpus then the resulting posterior inferences will be approximately calibrated with respect to that corpus. That is, posterior probabilities will represent frequencies across the corpus.
It is important to distinguish this frequentist Bayesian between-studies perspective from a more typically Bayes\-ian within-study framework view where posterior probabilities represent a study-specific model.

If our corpus-based prior distributions are to be used for default or routine Bayesian inference, then those inferences should not depend on trivial data transformations such as a change of the unit of measurement. Requiring our inference to be equivariant under linear transformations of the data greatly simplifies estimation of the prior (Theorem 2). Under this requirement, we only need to estimate the symmetric, marginal distribution of the observed  $z$-values. 

To use a corpus-based prior, one only needs to combine it with the (approximately) unbiased, normally distributed estimate of the parameter of interest and its standard error. This is a great advantage, because it allows anyone to perform a quick Bayesian re-analysis of a standard frequentist result. No need to wait for the author to do that!

\subsection{Limitations of our recommended approach}

The main difficulty of the method described in the present paper is the need to compile an honest corpus that is not affected by publication bias, file drawer effect, researcher degrees of freedom, fishing, forking paths, etc.  Promising sources are replication studies, registered reports or careful meta-analyses that make an effort to include also unpublished studies. 

A second caveat is our pragmatic requirement (\ref{requirement}) that our inference should be equivariant under linear transformations of the data. This requirement is important to ensure that it is not possible to manipulate the conclusions of a study by a change of measurement unit or by comparing group B to A instead of A to B. This requirement implies that $b/s$ and $s$ are independent and that the distribution of $b/s$ is symmetric. Those properties may or may not be reasonable in a particular corpus.

Our recommended approach makes use of the Bayesian formalism but is not fully Bayesian in that it does not correspond to any joint distribution of parameters and data.  Our choice of prior is improper, not in the sense of having no finite integral but in that it depends on the data (through the sample size $n$), which is not allowed in Bayesian inference.  For any particular dataset, the prior is proper, but the resulting inference violates Bayes' rule as new data come in.  For example, suppose an experiment with $n=100$ is analyzed using the methods described in the present article, and then the researcher goes and performs the study on 300 more people drawn from the same population.  We can treat this as new data and do Bayesian updating using the posterior from the experiment just performed as our prior for the analysis of the 300 new people, or we can consider the data as one experiment and go back to the default prior, this time scaled to $n=400$.  Because of the scaling of the prior, these two inferences will differ.

Arguably, however, some incoherence is appropriate for any default prior for a continuous parameter.  An informative prior for any regression coefficient will require some scaling \cite{gelman2008weakly}, and if this is not based on the data it would require an equivalent restriction to some class of appropriately-scaled problems.  The prior we have proposed here is unusual in that it is scaled to sample size, but this can be seen as a sort of rationalized version of current statistical practice which is to judge the plausibility of claims based on their $t$ ratio, the number of standard errors the estimate is from zero.

The availability of a corpus-based prior does not preclude using more specific prior information where available.  This can be considered as an implicit restriction of the corpus to a more relevant set of problems.

\subsection{A default default prior?}
In this article we have proposed a method for constructing a default prior for a class of problems by fitting a wide-tailed distribution ($t$ or mixture of normals) to data from a relevant corpus of careful studies.  But what about a truly default prior, to be applied in new problems, or settings where no reliable corpus is available, or for use in general-purpose software?  In this case we could see the virtue of a choice such as $t_1(0,1)$, which does a lot of shrinkage for noisy estimates but approaches the classical limit as the precision of the estimate increases, as illustrated in Figure \ref{fig:cauchy}.  There is no magic about this estimate, and it will be appropriate only to the extent that this prior reflects the distribution of underlying effects.  That said, we believe that this sort of standard-error-scaled prior can be a useful starting point in many settings.

\begin{figure}[htp] \centering{
\includegraphics[width=.5\textwidth]{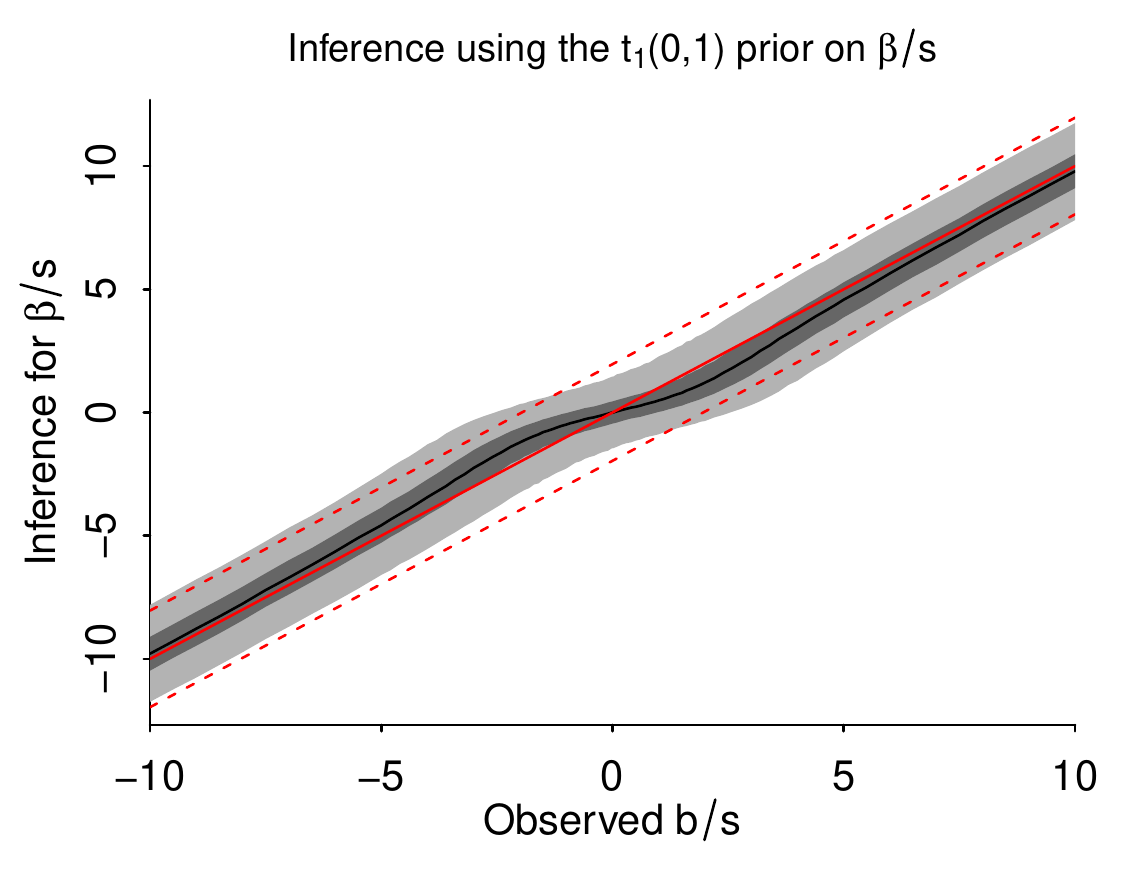}\includegraphics[width=.5\textwidth]{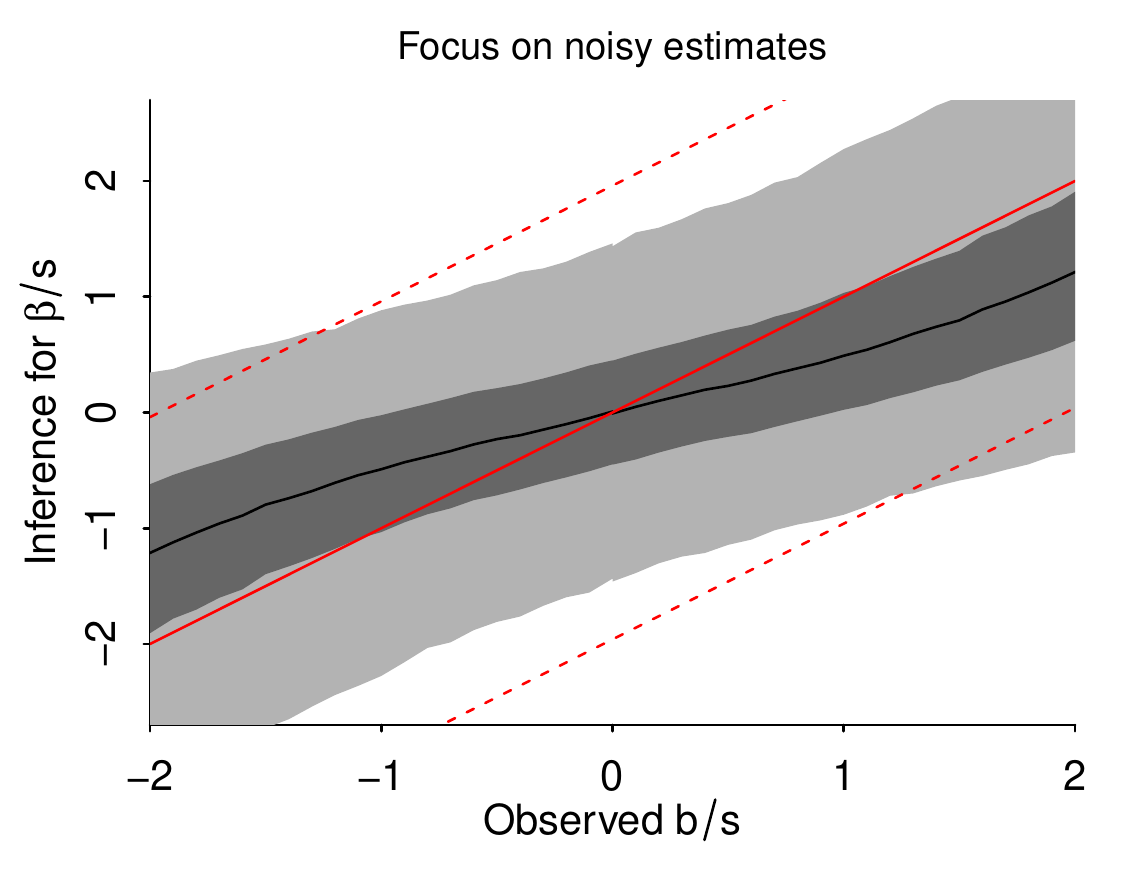}}
\caption{\em Inference given a default $t_1(0,1)$ prior on $b/s$.  Line and shaded bands show posterior median, 50\%, and 95\% intervals.  Red line and dotted lines show classical (or flat-prior) estimate and 95\% interval.  When the estimate is noisy (less than 2 standard errors away from zero, estimates are shrunk about halfway toward zero; as the precision of the estimate becomes higher, the inferences approach the classical limit.}\label{fig:cauchy}
\end{figure}

\bibliographystyle{plain}

\begin{thebibliography}{10}

\bibitem{berger2006case}
James~O. Berger.
\newblock The case for objective {Bayesian} analysis.
\newblock {\em Bayesian Analysis}, 1(3):385--402, 2006.

\bibitem{button2013power}
Katherine~S. Button, John P.~A. Ioannidis, Claire Mokrysz, Brian~A. Nosek,
  Jonathan Flint, Emma S.~J. Robinson, and Marcus~R. Munaf{\`o}.
\newblock Power failure: Why small sample size undermines the reliability of
  neuroscience.
\newblock {\em Nature Reviews Neuroscience}, 14(5):365, 2013.

\bibitem{carlin2010bayes}
Bradley~P. Carlin and Thomas~A. Louis.
\newblock {\em Bayes and Empirical Bayes Methods for Data Analysis}.
\newblock CRC Press, 2010.

\bibitem{open2015estimating}
Open~Science Collaboration.
\newblock Estimating the reproducibility of psychological science.
\newblock {\em Science}, 349(6251):aac4716, 2015.

\bibitem{dawid1973posterior}
A.~Philip Dawid.
\newblock Posterior expectations for large observations.
\newblock {\em Biometrika}, 60(3):664--667, 1973.

\bibitem{efron2012large}
Bradley Efron.
\newblock {\em Large-Scale Inference: Empirical Bayes Methods for Estimation,
  Testing, and Prediction}.
\newblock Cambridge University Press, 2012.

\bibitem{gelman2018anthropic}
Andrew Gelman.
\newblock The anthropic principle in statistics.
\newblock {\em Statistical Modeling, Causal Inference, and Social Science},
  2018.
\newblock
  \url{https://statmodeling.stat.columbia.edu/2018/05/23/anthropic-principle-statistics/}.

\bibitem{gelman2014beyond}
Andrew Gelman and John Carlin.
\newblock Beyond power calculations: Assessing type {S} (sign) and type {M}
  (magnitude) errors.
\newblock {\em Perspectives on Psychological Science}, 9(6):641--651, 2014.

\bibitem{gelman2017beyond}
Andrew Gelman and Christian Hennig.
\newblock Beyond subjective and objective in statistics (with discussion).
\newblock {\em Journal of the Royal Statistical Society, Series A}, 2017.

\bibitem{gelman2008weakly}
Andrew Gelman, Aleks Jakulin, Maria~Grazia Pittau, and Yu-Sung Su.
\newblock A weakly informative default prior distribution for logistic and
  other regression models.
\newblock {\em Annals of Applied Statistics}, 2(4):1360--1383, 2008.

\bibitem{gelman2013garden}
Andrew Gelman and Eric Loken.
\newblock The statistical crisis in science.
\newblock {\em American Scientist}, 102:460--465, 2014.

\bibitem{gelman2000type}
Andrew Gelman and Francis Tuerlinckx.
\newblock Type {S} error rates for classical and {Bayesian} single and multiple
  comparison procedures.
\newblock {\em Computational Statistics}, 15(3):373--390, 2000.

\bibitem{ghosh2011objective}
Malay Ghosh et~al.
\newblock Objective priors: An introduction for frequentists.
\newblock {\em Statistical Science}, 26(2):187--202, 2011.

\bibitem{greenland2015penalization}
Sander Greenland and Mohammad~Ali Mansournia.
\newblock Penalization, bias reduction, and default priors in logistic and
  related categorical and survival regressions.
\newblock {\em Statistics in Medicine}, 34(23):3133--3143, 2015.

\bibitem{greenland2016statistical}
Sander Greenland, Stephen~J. Senn, Kenneth~J. Rothman, John~B. Carlin, Charles
  Poole, Steven~N. Goodman, and Douglas~G. Altman.
\newblock Statistical tests, {P} values, confidence intervals, and power: A
  guide to misinterpretations.
\newblock {\em European Journal of Epidemiology}, 31(4):337--350, 2016.

\bibitem{ioannidis2005contradicted}
John P.~A. Ioannidis.
\newblock Contradicted and initially stronger effects in highly cited clinical
  research.
\newblock {\em Journal of the American Medical Association}, 294(2):218--228,
  2005.

\bibitem{ioannidis2008most}
John P.~A. Ioannidis.
\newblock Why most discovered true associations are inflated.
\newblock {\em Epidemiology}, 19(5):640--648, 2008.

\bibitem{kass1996selection}
Robert~E. Kass and Larry Wasserman.
\newblock The selection of prior distributions by formal rules.
\newblock {\em Journal of the American Statistical Association},
  91(435):1343--1370, 1996.

\bibitem{o2012bayesian}
Anthony O'Hagan and Luis Pericchi.
\newblock Bayesian heavy-tailed models and conflict resolution: A review.
\newblock {\em Brazilian Journal of Probability and Statistics},
  26(4):372--401, 2012.

\bibitem{rothstein2006publication}
Hannah~R. Rothstein, Alexander~J. Sutton, and Michael Borenstein.
\newblock {\em Publication Bias in Meta-Analysis: Prevention, Assessment and
  Adjustments}.
\newblock Wiley, 2006.

\bibitem{rubin1984bayesianly}
Donald~B. Rubin.
\newblock Bayesianly justifiable and relevant frequency calculations for the
  applied statistician.
\newblock {\em Annals of Statistics}, 12(4):1151--1172, 1984.

\bibitem{stefanski2002calculus}
Leonard~A. Stefanski and Dennis~D. Boos.
\newblock The calculus of {M}-estimation.
\newblock {\em American Statistician}, 56(1):29--38, 2002.

\bibitem{vanZwetCator}
E.~W. van Zwet and E.~A. Cator.
\newblock The winner's curse and the need to shrink.
\newblock 2021.
\newblock \url{http://arxiv.org/abs/2009.09440/}.

\bibitem{vanZwetSenn}
E.~W. van Zwet, S.~Schwab, and S.~J. Senn.
\newblock The statistical properties of {RCTs}.
\newblock 2021.

\bibitem{van2018default}
Erik van Zwet.
\newblock A default prior for regression coefficients.
\newblock {\em Statistical Methods in Medical Research}, 28:3799--3807, 2019.

\end{thebibliography}

\appendix

\section{Proof of Theorem 2}

\begin{lemma}
Suppose $X$ and $Y$ are random variables and $Y$ is positive. Then
\begin{equation}\label{premise}
f_{X|Y=y}(x) = c f_{X|Y=cy}(cx) 
\end{equation}
for every $c>0$, if and only if $Z=X/Y$ and $Y$ are independent.
\end{lemma}

\begin{proof}
Note that
\begin{equation}\label{l:1}
f_{Z \mid Y=y}(x/y) = y f_{X \mid Y=y}(x).
\end{equation}
Now, suppose (\ref{premise}) holds. If we choose $c=1/y$ then
\begin{equation}\label{l:2}
f_{X|Y=y}(x) = \frac{1}{y} f_{X|Y=1}(x/y) = \frac{1}{y} f_{Z |Y=1}(x/y).
\end{equation}
Combining (\ref{l:2}) with (\ref{l:1}), we find
\begin{equation}\label{key}
f_{Z \mid Y=y}(x/y) =  f_{Z |Y=1}(x/y).
\end{equation}
Since this equality holds for every $x$ and $y>0$, we see that $Z$ and $Y$ are independent.

\medskip \noindent
To prove the converse,  independence of $Z$ and $Y$ implies
\begin{equation}\label{l:3}
f_{Z|Y=y}(x/y)=f_{Z|Y=cy}(x/y)=cyf_{X|Y=cy}(cx).
\end{equation}
Combining (\ref{l:3}) with (\ref{l:1}), we find  (\ref{premise}).
\end{proof}

\medskip \noindent
We recall Theorem 2 and prove it.

\setcounter{theorem}{1}
\begin{theorem}
Requirement (\ref{requirement}) holds if and only if
\begin{enumerate}
\item $s$ and $\beta/s$ are independent.
\item The distribution of $\beta/s$ is symmetric around zero.
\end{enumerate}
\end{theorem}

\begin{proof}
Conditionally on $\beta$ and $s$, $b$ has the normal distribution with mean $\beta$ and standard deviation $s$,
$$f(b \mid \beta,s) = \frac{1}{\sqrt{2\pi\, s^2}} e^{-\frac{(b - \beta)^2}{2s^2}},$$
It is easy to verify that
\begin{equation}
f(b \mid \beta, s) =|c| f(c b \mid c\beta, |c|s).
\end{equation}
for every $c \neq 0$. So, requirement (\ref{requirement}) is equivalent with
\begin{equation} \label{p:1}
f(\beta \mid s)=|c| f(c\beta \mid |c|s),
\end{equation}
for every $c \neq 0$. Now, by lemma 1 we have that (\ref{p:1}) holds for all $c>0$ if and only if $s$ and $\beta/s$ are independent. Taking $c=-1$  in (\ref{p:1}), we see that the conditional distribution of $\beta$ given $s$ is symmetric around zero. It follows that the marginal distribution of $\beta/s$ is also symmetric. 
\end{proof}

\section{Details of section 4.1}
We model the distribution of the SNR as a mixture of two zero-mean normal distributions. We the mixture model in terms of a mixture indicator $d$ that can take on the values 1 or 2:
\begin{itemize}
\item $\mbox{Pr}(d=1)=1-\mbox{Pr}(d=2)=p$.
\item Given $d=i$, $\beta/s$ has the normal distribution with mean 0 and standard deviation $\tau_i$ ($i=1,2$).
\item Given $d$ and $\beta/s$, $z$ has the normal distribution with mean $\beta/s$ and standard deviation 1.
\end{itemize}
The marginal distribution $g$ of $z$ is a mixture of two zero-mean normals with standard deviations $\sqrt{\tau_1^2+1}$ and $\sqrt{\tau_2^2+1}$, and mixing proportion $p$,
\begin{equation}
g(z) = \frac{p}{\sqrt{\tau_1^2+1}} \, \varphi \left( \frac{z}{\sqrt{\tau_1^2+1}} \right)  +  \frac{1-p}{\sqrt{\tau_2^2+1}} \, \varphi \left( \frac{z}{\sqrt{\tau_2^2+1}} \right), 
\end{equation}
where $\varphi$ is the standard normal density function. We can easily fit this model from a sample $z_1,z_2,\dots,z_n$, using maximum likelihood or Bayesian inference. Upon fitting this model, we can obtain the conditional distribution of $\beta$ given $b$ and $s$. This posterior distribution is again a mixture of two normals. The mixing proportion is
\begin{align}
\mbox{Pr}(d=1 \mid b,s) = \frac{p}{\sqrt{\tau_1^2+1}} \varphi \left( \frac{b}{s \sqrt{\tau_1^2+1}} \right) \frac{1}{g(b/s)}
\end{align}
and the two normals have  means
\begin{equation}
\mu_i = s\,  \E(\beta/s \mid b, s, d=i) = \frac{b\, \tau_i^2}{\tau_i^2 +1}
\end{equation}
and variances
\begin{equation}
\sigma_i^2 = s^2\, {\rm Var}(\beta/s \mid b, s, d=i) = \frac{s^2\, \tau_i^2}{\tau_i^2 +1}
\end{equation}
for $i=1,2$. We define the ``shrinkage factor'' as $b/ \E(\beta \mid b, s)$, which depends on $b$ and $s$ only through $z=b/s$.

\end{document}